\documentclass[journal]{IEEEtran}

\usepackage[dvipdf]{graphicx,color}
\usepackage{epsfig}
\usepackage{graphpap}
\usepackage{amssymb}
\usepackage[cmex10]{amsmath}
\usepackage{amsfonts}
\usepackage{amsmath}
\usepackage{multirow}
\usepackage{algorithm}
\usepackage{algorithmicx}
\usepackage{algpseudocode}
\usepackage{tabularx}
\usepackage{color}
\usepackage{xcolor}
\usepackage{varwidth}
\usepackage{multicol}
\usepackage{subfigure}
\usepackage{enumerate}
\usepackage{amsthm}
\usepackage{indentfirst}
\usepackage{exscale}
\usepackage[font=footnotesize]{caption}
\usepackage{stfloats}
\usepackage{float}
\usepackage{epstopdf}
\usepackage[utf8]{inputenc}
\usepackage[english]{babel}

\makeatletter
\newcommand{\Statey}{\Statex\hspace*{\ALG@thistlm}}
\makeatother
\makeatletter
\newcommand{\multiline}[1]{%
  \begin{tabularx}{\dimexpr\linewidth-\ALG@thistlm}[t]{@{}X@{}}
    #1
  \end{tabularx}
}
\def\underbracex#1#2{\mathop{\vtop{\m@th\ialign{##\crcr
   $\hfil\displaystyle{#2}\hfil$\crcr
   \noalign{\kern3\p@\nointerlineskip}%
   #1\crcr\noalign{\kern3\p@}}}}\limits}

\def\underbracea{\underbracex\upbracefilla}

\def\upbracefilla{$\m@th \setbox\z@\hbox{$\braceld$}%
  \bracelu\leaders\vrule \@height\ht\z@ \@depth\z@\hfill 
\kern\p@\vrule \@width\p@\kern\p@\vrule \@width\p@\kern\p@\vrule \@width\p@
$}

\def\upbracefillb{$\m@th \setbox\z@\hbox{$\braceld$}%
\vrule \@width\p@\kern\p@\vrule \@width\p@\kern\p@\vrule \@width\p@\kern\p@
 \leaders\vrule \@height\ht\z@ \@depth\z@\hfill\bracerd
  \braceld\leaders\vrule \@height\ht\z@ \@depth\z@\hfill
\kern\p@\vrule \@width\p@\kern\p@\vrule \@width\p@\kern\p@\vrule \@width\p@
$}

\def\upbracefillc{$\m@th \setbox\z@\hbox{$\braceld$}%
\vrule \@width\p@\kern\p@\vrule \@width\p@\kern\p@\vrule \@width\p@\kern\p@
\leaders\vrule \@height\ht\z@ \@depth\z@\hfill
\kern\p@\vrule \@width\p@\kern\p@\vrule \@width\p@\kern\p@\vrule \@width\p@
$}

\def\upbracefilld{$\m@th \setbox\z@\hbox{$\braceld$}%
\vrule \@width\p@\kern\p@\vrule \@width\p@\kern\p@\vrule \@width\p@\kern\p@
 \leaders\vrule \@height\ht\z@ \@depth\z@\hfill\braceru$}

 \def\underbracebd{\underbracex\upbracefillbd}
\def\upbracefillbd{$\m@th \setbox\z@\hbox{$\braceld$}%
\vrule \@width\p@\kern\p@\vrule \@width\p@\kern\p@\vrule \@width\p@\kern\p@
\bracerd\braceld
 \leaders\vrule \@height\ht\z@ \@depth\z@\hfill\braceru$}
\makeatother

\newtheorem{theorem}{Theorem}
\newtheorem{lemma}{Lemma}

\newtheorem{remark}{Remark}

\newtheorem{assumption}{Assumption}

\begin{document}

%
\title{Energy-efficient Federated Learning for UAV Communications}


\author{\IEEEauthorblockN{Chien-Wei Fu, {\textit{Student Member}}, {\textit{IEEE}}}, Meng-Lin Ku, {\textit{Senior Member}}, {\textit{IEEE}}}

\maketitle
{
\begin{abstract}
In this paper, we propose an unmanned aerial vehicle (UAV)-assisted federated learning (FL) framework that jointly optimizes UAV trajectory, user participation, power allocation, and data volume control to minimize overall system energy consumption. We begin by deriving the convergence accuracy of the FL model under multiple local updates, enabling a theoretical understanding of how user participation and data volume affect FL learning performance. The resulting joint optimization problem is non-convex; to address this, we employ alternating optimization (AO) and successive convex approximation (SCA) techniques to convexify the non-convex constraints, leading to the design of an iterative energy consumption optimization (ECO) algorithm. Simulation results confirm that ECO consistently outperform existing baseline schemes.

\end{abstract}

}
\section{Introduction}
{}

The rapid rise of data-driven applications in 6G networks, such as smart cities and autonomous systems, calls for scalable and privacy-preserving machine learning solutions  {\cite{Wang2023}}. Traditional centralized learning is inefficient and raises privacy concerns due to large-scale data transmission {\cite{McMahan2017}}. Federated Learning (FL) addresses this by enabling local training on devices and only sharing model parameters, reducing communication costs and preserving data privacy {\cite{Nguyen2021}}.

Meanwhile, Unmanned Aerial Vehicles (UAVs) have become essential in extending network coverage thanks to their high mobility and flexible deployment \cite{Pham2022b}. Integrating UAVs with FL enhances edge computing capabilities, improves wireless links, and boosts training efficiency through optimized flight trajectories. However, UAV-assisted FL also introduces challenges, such as efficient resource allocation, reliable communication, trajectory planning, and energy management.

While many studies have contributed to UAV-assisted FL, they often exhibit limited scope. Works such as \cite{Pham2022} and \cite{Pham2021} overlook user selection and data volume control. Although \cite{Chen2024} includes user selection and model compression, it omits data volume control. Studies like \cite{He2020} and \cite{Li2024} consider learning efficiency or latency but ignore energy consumption and holistic optimization. Many, including  \cite{Pham2022}, \cite{Pham2021}, \cite{Chen2024} and \cite{Zhagypar2025}, assume static UAVs, neglecting trajectory and energy dynamics. While \cite{Zheng2024} addresses energy-efficient client selection and power control, it lacks joint optimization with data volume or trajectory. 
{}

{}

{}
To address the above limitations, this paper proposes a UAV-assisted FL framework that jointly optimizes UAV trajectory, user selection, power control, and data volume allocation to minimize system energy consumption. A theoretical convergence analysis is provided, revealing the impact of user selection and data volume on model accuracy. We develop an analytical model linking data volume to FL performance and propose a system model capturing UAV-user interactions. The non-convex optimization problem is solved using successive convex approximation (SCA) with an alternating optimization (AO) algorithm. Simulation results demonstrate the effectiveness of the proposed approach across different model sizes, achieving lower energy consumption compared to baseline methods.
\begin{figure}[h]
\centering
\includegraphics[width=0.48\textwidth]{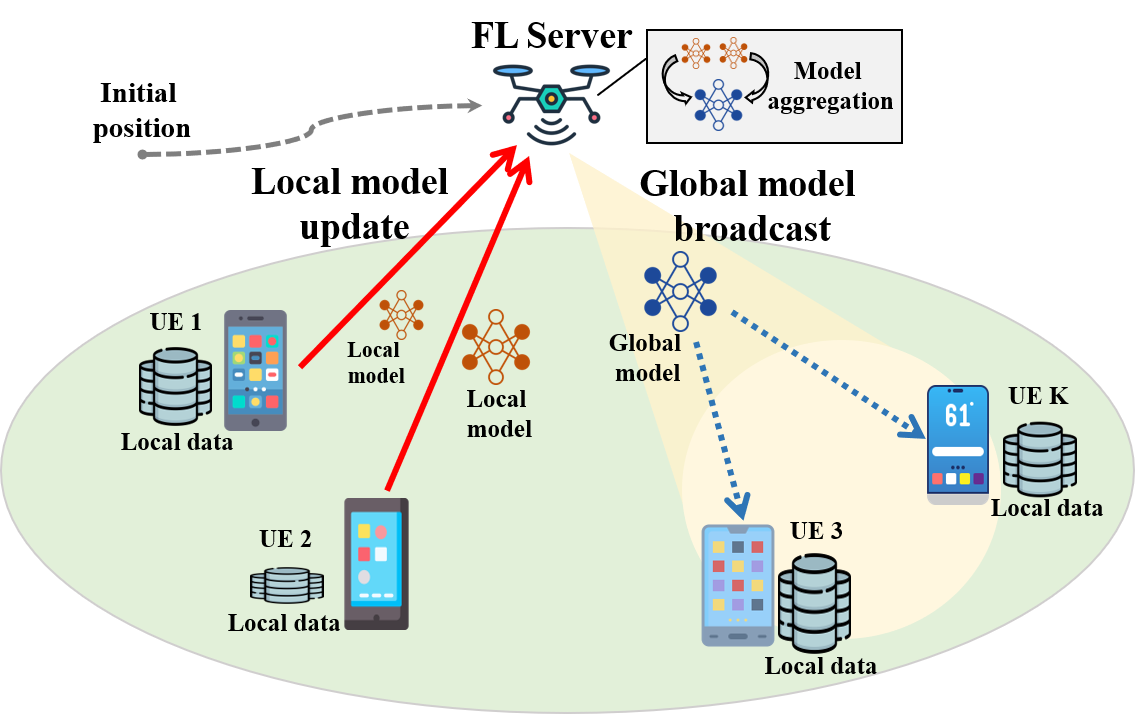}
\caption{UAV-assisted federated learning communications ($K=4$).}
\label{Fig_1}
\end{figure}

\section{System Model And Problem Formulation}\label{System_model}

\subsection{System Model}
Fig. \ref{Fig_1} shows a UAV-assisted FL communication network, consisting of a UAV and $K$ ground users (UEs) with different sizes of local data. The UAV operates as an FL server and flies over the area to perform FL simultaneously with the $K$ UEs. We adopt a time-slotted model, assuming that the entire task period $T$ is divided into $N$ time slots, with $N$ discrete time instants ($n=1,\ldots,N$). We define the sets of UEs and time instants as $\mathcal{K}=\{1,\ldots,K\}$  and $\mathcal{N}=\{1,\ldots,N-1\}$, respectively. The time intervals are detailed as in Fig. \ref{Fig_timeslot}, where $t^{cp}_k[n]$ represents the local model computation time of the $k$th UE at time $n$. Furthermore, $ t^{fly}[n]$ and $t^{hov}[n]$ denote the flight and hovering time of the UAV at time $n$, and $t^{cm}$, $t^{agg}$, and $t^{bc}$ represent the time for UEs to upload the FL model, for the model to be aggregated on the UAV, and for the UAV to broadcast the FL model, respectively. To enable synchronized FL among UEs, the time duration of $t^{cm}$, $t^{agg}$, and $t^{bc}$ is assumed to be fixed.
\begin{figure}[h]
\centering
\includegraphics[width=0.48\textwidth]{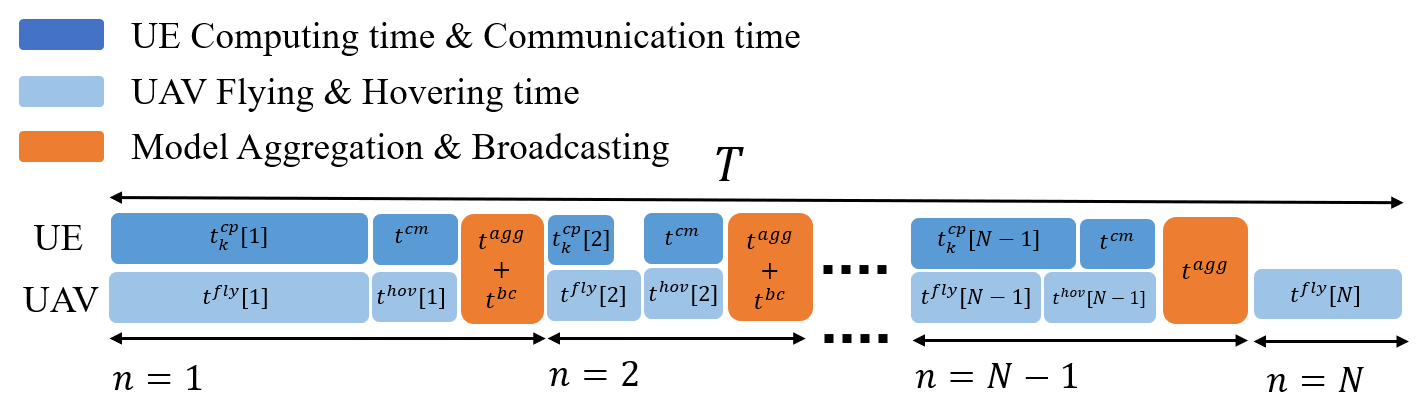}
\caption{Time slot model.}
\label{Fig_timeslot}
\end{figure}

We assume the UEs' positions remain unchanged, with the $k$th UE's horizontal two-dimensional (2D) coordinate given by
\begin{align}
&\textbf g_{UE,k}=[{\bar{x}}_k,{\bar{y}}_k]^T \in \mathbb{R}^2 , \forall k\in \mathcal{K}\label{eq2} \,.
\end{align}The UAV is assumed to fly at a fixed altitude $H$ and constant speed $v_{UAV}$, with its horizontal 2D coordinate at time $n$:
\begin{align}\label{eq3}
&\textbf q[n]=[x[n],y[n]]^T \in \mathbb{R}^{2},\forall n \in\mathcal{N}\cup\{0,N\} \,,
\end{align}
and the UAV's trajectories are subject to the constraints of the initial position $\textbf q^{ini}$ and final position $\textbf q^{fin}$:
\begin{align}
&\textbf q[0]=\textbf q^{ini} ;\label{eq_uav_ini}\\ 
&\textbf q[N]=\textbf q^{fin} .\label{eq_uav_fin}
\end{align}

For the channel model, we assume the UAV flies at a high enough altitude, ensuring a line-of-sight (LOS) channel between the UAV and UEs. The path loss (in decibels) between the UAV and the $k$th UE at time $n$ is given as {\cite{Mozaffari2016}}:
\begin{align}
& g_k\left[n\right]=20{\log}_{10}\left(\frac{4\pi f_{c}d_{k}\left[n\right]}{c}\right) \,, \forall k\in\mathcal{K}, \forall n\in\mathcal{N},\,\label{eq_pathloss}
\end{align}where $f_{c}$ is the carrier frequency (Hz), $c$ is the speed of light (m/s), and $d_{k}[n]$ represents the distance between the UAV and the $k$th UE at time $n$, given as
\begin{align}
d_{k}\left[n\right]= \sqrt{\left\|\textbf q[n]-\textbf g_{UE,k}\right\|_2^2+H^2} \,, \forall k\in\mathcal{K}, \forall n\in\mathcal{N}.\label{eq_distance_uav_UE}
\end{align}

The time-slotted model is detailed as follows. Let $D_k$ denote the data amount used by the $k$th UE for FL participation. The local computation time for UE $k$ at time $n$ is calculated as:
\begin{align}\label{time_cp_user}
    &t_k^{cp}[n]=a_k[n]D_k\Phi_k, \forall k\in\mathcal{K}, \forall n\in\mathcal{N},
\end{align}
where $\Phi_k=I(\frac{C}{f_{cpu,k}})$ with $f_{cpu,k}$ denoting the CPU frequency, $I$ is the number of local update iterations, and $C$ is the computation required to process one bit. The binary variable $a_k[n] \in \{0,1\}$ indicates FL participation ($a_k[n] = 1$ if the $k$th UE participates at time $n$, otherwise $a_k[n] = 0$). 

According to the time-slotted model in Fig. \ref{Fig_timeslot}, the UE performs the local model computation while the UAV is in flight, and the UE transmits the local model during the UAV hovering phase. In this scenario, the length of each time instant can vary, subject to the following constraints:
\begin{align}
t^{cm} \leq t^{hov}[n] ,\forall n \in \mathcal{N} ;\label{cons_communacation_hovering}
\end{align}
\begin{align}
t^{cm}+t_k^{cp}[n] \leq t^{fly}[n]+t^{hov}[n], \forall k\in \mathcal{K},\forall n\in \mathcal{N};\label{cons_time_UE_less_UAV}
\end{align}
\begin{align}
&\sum^N_{n=1}t^{fly}[n]+\sum^{N-1}_{n=1}(t^{hov}[n]+t^{agg})+(N-2) t^{bc}\leq T,\label{cons_total_time}
\end{align}
where the flight time of the UAV at time $n$ is given by
\begin{align}
    t^{fly}[n]=\frac{\|\textbf{q}[n]-\textbf{q}[n-1]\|}{v_{UAV}},\forall n \in\mathcal{N}\cup\{N\}.
\end{align}
The constraint (\ref{cons_communacation_hovering}) mandates the UAV to hover during FL model uploading to ensure a stable communication link. To perform FL model aggregation, the constraint (\ref{cons_time_UE_less_UAV}) ensures that the UE's local model computation and communication time does not exceed the UAV's flight and hover time. The constraint (\ref{cons_total_time}) ensures that the UAV's flight and hover time, along with the FL model aggregation time $t^{cm}$ and global model broadcasting time $t^{bc}$, do not exceed the task period $T$.

\subsection{Transmission Model}
Since the UEs share the same bandwidth $W$ for communication, the UAV experiences multiuser interference. The signal-to-interference-plus-noise ratio (SINR) for the $k$th UE at the UAV and time $n$ is expressed as
\begin{align}
{\Gamma}_{k}\left[n\right]=\frac{p_k\left[n\right]\tilde g_{k}\left[n\right]}{\sum_{i=1,i\neq k}^{K}p_i\left[n\right] \tilde g_{i}\left[n\right]+\sigma_z^2} , \forall k \in \mathcal{K},\forall n \in \mathcal{N}, \label{eq_SINR}
\end{align}
{where ${\tilde{g}}_{k}\left[n\right]=10^{\frac{-g_k[n]}{10}}$} , {$\sigma_z^2$} is the power of additive white Gaussian noise, and $p_k\left[n\right]$ is the uplink transmit power of UEs, subject to the constraint:
\begin{align} \label{UE_power}
    0\leq p_k[n]\leq p_{UE}^{max},\forall k \in \mathcal{K},\forall n\in \mathcal{N}, 
\end{align}
where $p_{UE}^{max}$ is the maximum allowable power of UEs. The achievable rate of the $k$th UE at time $n$ can be calculated as
\begin{align}\label{eq_transmi_rate}
R_k[n]= {W{\log}_2\left(1+{\Gamma}_{k}\left[n\right]\right)},\forall k\in\mathcal{K},\forall n\in\mathcal{N}.
\end{align}

Assuming the local model size is $\mathcal{Q}$, a rate constraint is introduced to ensure that the uplink data volume transmitted by the $k$th UE at time $n$ within the communication duration $t^{cm}$ is no less than $\mathcal{Q}$, thereby enabling model aggregation:
\begin{align}\label{UE_tran_cons}
   a_k[n]\mathcal{Q}\leq t^{cm}R_k[n],\forall k \in \mathcal{K},\forall n\in\mathcal{N}.
\end{align}

After the FL server generates the global model via aggregation, it broadcasts the model to the FL users selected for participation at the subsequent time instant. Notably, the set of participating users at time $n$ may differ from that at the previous time step. The broadcast rate for user $k$ at time $n$, denoted as $R^{bc}_k[n]$, is defined as follows:
\begin{align}
    R^{bc}_k[n]&= {W{\log}_2\left(1+\frac{p_{UAV}[n]\tilde g_k[n]}{\sigma_z^2}\right)}\nonumber\\
    &\qquad\qquad\qquad,\forall k\in\mathcal{K},\forall n\in \mathcal{N}\setminus\{N-1\},
\end{align}
where $p_{UAV}[n]$ is the broadcast power of the UAV, constrained by its maximum power limit $p_{UAV}^{max}$:
\begin{align}\label{UAV_power}
    0\leq p_{UAV}[n]\leq p_{UAV}^{max},\forall n\in \mathcal{N}\setminus\{N-1\}.
\end{align}
To ensure that the global model broadcasting is completed within the designated broadcasting time $t^{bc}$, the following constraint is imposed:
\begin{align}\label{UAV_broad_cons}
    a_k[n+1]\mathcal{Q}\leq t^{bc}R_k^{bc}[n],\forall k \in \mathcal{K},\forall n\in \mathcal{N}\setminus\{N-1\}.
\end{align}

\subsection{FL Model}
A global loss function $f_G(\cdot)$ is defined to serve as a performance measure for the FL model $f_G(\mathbf{{w}})\triangleq\sum_{k=1}^{K}f_{L,k}\left(\mathbf{w}\right)$, {where $f_{L,k}(\mathbf{w})$ is the average local loss function of the $k$th UE, evaluated over its local dataset of size $D_k$, and $\mathbf{w}$ is the FL model parameters. Let $\mathcal{I} = \{nI \mid n \in \mathcal{N}\}$ be the set of global aggregation intervals. The participating UEs remains fixed within each time slot, during which every participating UE performs $I-1$ local updates, followed by one global aggregation at the UAV server. The participating index of the $k$th UE at the $i$th FL model update is thus given by:
\begin{align}
   a_k^i=a_k[n] \in\{0,1\}, \forall i\in \{nI-1,nI-2,...,nI-I\}.\label{ak_i_to_n}
\end{align}
The FL model parameter $\mathbf{w}$ is collaboratively trained by the UEs with assistance from the UAV server, and the update process is divided into two steps:
\subsubsection{Local Update}
Denote the local model parameters of the $k$th UE in the $i$th update as $\mathbf{w}_k^{i}$. Each UE performs local updates utilizing stochastic gradient descent (SGD) as follows:
\begin{align}
\mathbf{w}_k^{i}=\mathbf{w}_k^{i-1}-\eta\nabla f_{L,k}(\mathbf{w}_k^{i-1},s_k^{i-1}), \forall i\notin \mathcal{I}\, ,\label{local_update}
\end{align}
where $\eta$ is the learning rate, $\nabla f_{L,k}$ is the gradient of $f_{L,k}(\mathbf{w})$, and $s_k^{i-1}$ represents the random uniform sampling of data used by the $k$th UE in the $(i-1)$th SGD update.
\subsubsection{Global Aggregation}
The global aggregation is performed at the UAV server after every local updates, occurring at each time instant $nI$:
\begin{align}
        \mathbf{w}_k^{i+1}&=\sum_{k=1}^K\frac{ a_k^i D_k}{\sum_{j=1}^K a_j^i D_j}\mathbf{w}_k^{i}, \forall i+1\in \mathcal{I}.\label{gb_agg}
\end{align}

To facilitate the convergence analysis of the FL model, a "virtual" model aggregation is introduced as follows:  
\begin{align}
\bar{\mathbf{w}}^{i} = \sum_{k=1}^K \frac{a_k^i D_k}{\sum_{j=1}^K a_j^i D_j} \mathbf{w}_k^{i}, \forall i,
\end{align}
where $\bar{\mathbf{w}}^{i}$ is the virtual global model obtained from the local models in the $i$th update. The term "virtual" indicates that this aggregation is not physically performed on the UAV server, but serves as a conceptual representation for analysis. Using its recursive form, the sequence of virtual updates is given as:  
\begin{align} 
\bar{\mathbf{w}}^{i+1} = \bar{\mathbf{w}}^{i} - \eta \sum_{k=1}^K \frac{a_k^i D_k}{\sum_{j=1}^K a_j^i D_j} \nabla f_{L,k}(\mathbf{w}_k^{i}, s_k^i). \label{virtual_update}
\end{align}
This recursive formulation shows that each local update contributes to an equivalent global model update, thereby linking local and global updates. Below, we introduce common assumptions about the loss function $f(\cdot)$ {\cite{Cao2022}}, where $f(\cdot)$ may represent either the global loss $f_{G}(\cdot)$ or local loss $f_{L,k}(\cdot)$.
Below, we introduce common assumptions about the loss function $f(\cdot)$ {\cite{Cao2022}}, where $f(\cdot)$ may represent either the global loss $f_{G}(\cdot)$ or local loss $f_{L,k}(\cdot)$.

\begin{assumption}\label{assumption_uconcave}
($\mu$-strongly convex). For all $\textbf{\rm a},\textbf{\rm b}$, $f(\textbf{\rm b})\geq f(\textbf{\rm a})+\langle\nabla f(\textbf{\rm a}),\textbf{\rm b}-\textbf{\rm a}\rangle+\frac{\mu}{2}\|\textbf{\rm b}-\textbf{\rm a}\|^2$, {where $\mu>0$ is a strongly convex parameter.}
\end{assumption}

\begin{assumption}\label{assumption_lsmooth}
(L-smooth). For all $\textbf{\rm a},\textbf{\rm b}$, $\|\nabla f(\textbf{\rm a})-\nabla f(\textbf{\rm b})\|\leq L\|\textbf{\rm a}-\textbf{\rm b}\|,$ {where $L>0$ is a smoothness parameter.}
\end{assumption}

\begin{assumption}\label{Assumption_sigmak}
(Bounded sample variance in stochastic gradients). For any FL user in the $i$th SGD update, $\mathbb{E} \Big[\left\|\nabla f_{L,k}(\mathbf{w}_k^i)-\nabla f_{L,k}(\mathbf{w}_k^i,s^i_k)\right\|^2 \Big]\leq\epsilon^2_v,$ where {$\epsilon^2_v$ is an upper bound of the gradient variation.}
\end{assumption}

\begin{assumption}\label{Assumption_sigmaG}
(Bounded square norm expectation in stochastic gradients). For any FL user in the $i$th SGD update, $\mathbb{E} \Big[\left\|\nabla f_{L,k}(\mathbf{w}_k^i,s^i_k)\right\|^2 \Big]\leq\epsilon_s^2,$ where   {$\epsilon_s^2$ is an upper bound of the update magnitude.}
\end{assumption}
{\begin{remark}\label{assumption_lsmooth_remark_1}
The L-smooth property (Assumption \ref{assumption_lsmooth}) implies the following inequality: For all $\textbf{\rm a},\textbf{\rm b}$, $\|\nabla f(\textbf{\rm a})\|^2\leq2 L\left(f(\textbf{\rm a})-f(\textbf{\rm b})\right)$ \cite{gower2018}.
\end{remark}
\begin{remark}\label{assumption_lsmooth_remark_2}
    Assumption \ref{assumption_lsmooth} can be extended to the inequality: For all ${\rm x}\in\mathbb{R}^d$, $f({{\rm x}})-f({{\rm x}^*})\leq\frac{L}{2}\left\|{{\rm x}}-{{\rm x}^*}\right\|^2$, where ${\rm x}^*=\arg \min_{{\rm x}\in\mathbb{R}^d} f({\rm x})$ \cite{gower2018}.
\end{remark}
}
\begin{theorem}\label{model_acc_Theorem1}
    Let $\mathbf{w}^*$ be the optimal FL parameter of the global function. Under Assumptions 1-4 and a learning rate $\eta\leq \frac{1}{2L}$, the FL model accuracy with $(i+1)$ updates is bounded by 
    \begin{align}
        &\mathbb{E}\Big[f_G(\bar{\mathbf{w}}^{i+1})-f_G(\mathbf{w}^*)\Big]\nonumber\\
        &\leq\frac{L}{2}\Biggl[\omega^{i+1}\mathbb{E}\left[\|\bar{\mathbf{w}}^0-\mathbf{w}^*\|^2\right]+A_1\left(\frac{1-\omega^{i+1}}{\eta\mu}\right)\nonumber\\
        &\;+\eta^2\sum_{l=0}^{i}\left(\omega^{i-l}\sum_{k=1}^K(\bar D_k^{l})^2\epsilon^2_v\right)\Biggr],\label{Theo1_acc}
    \end{align}
    where {
     $\omega=1-\eta\mu,
     A_1=(1+\frac{\zeta}{2\eta})I^2\eta^2\epsilon_s^2+\frac{\eta L^2 \epsilon_w}{2}(\zeta+4\eta),
     \zeta=2\eta(1-\eta2L)$, $ \bar D_k^i=\frac{a_k^i D_k}{\sum_{j=1}^Ka_j^i D_j},$ and $\|\mathbf{w}^*-\mathbf{w}_k^*\|^2\leq\epsilon_w$.
        }
\end{theorem}
\begin{proof}
    See Appendix \ref{Theorem1_proof} for the detailed proof.
\end{proof}Theorem \ref{model_acc_Theorem1} shows that FL model accuracy exhibits a quadratic relationship with $\bar{D}_k^i$, highlighting the critical role of UE data volume and participation in determining accuracy. By applying Theorem \ref{model_acc_Theorem1}, we introduce an FL model accuracy constraint with a threshold $\epsilon_G$ to ensure convergence after $(N-1)I$ model updates (at the end of the task):
\begin{align}\label{model_accuracy_cons}
    &\frac{L}{2}\Biggl[\omega^{(N-1)I}\mathbb{E}\left[\|\bar{\mathbf{w}}^0-\mathbf{w}^*\|^2\right]+A_1\left(\frac{1-\omega^{(N-1)I}}{\eta\mu}\right)\nonumber\\
    &\;+\eta^2\sum_{n=1}^{N-1}\left(\omega^{((N-1)-n)I}\left(\frac{1-\omega^I}{\eta\mu}\right)\sum_{k=1}^K{(\bar D_k[n])^2}\epsilon^2_v\right)\Biggr]\nonumber\\
    &\leq\epsilon_G \, ,
\end{align}
{where $\bar D_k[n]=\frac{a_k[n] D_k}{\sum_{j=1}^Ka_j[n] D_j}$ is derived from $\bar D_k^i$ via (\ref{ak_i_to_n}).}

To ensure adequacy and heterogeneity of learning data, the total data volume contributed by all UEs during the task period $T$ must exceed a predefined threshold $D_{th}$, given by:  
\begin{align}\label{tot_data_size}
    \sum_{n=1}^{N-1} a_k[n]D_k\geq D_{th},\forall k \in \mathcal{K}.
\end{align}
Additionally, we assume that FL training occurs in every time slot. To maintain adequate decentralization and mitigate the influence of individual devices on model convergence, the number of participating FL users per time slot must meet a minimum requirement $a_{min}$, given as:
\begin{align}\label{tot_UE_par}
    \sum_{k=1}^{K} a_k[n]\geq a_{min},\forall n \in \mathcal{N}.
\end{align}

\subsection{Energy Consumption Model}
The energy consumption includes communication and computation. The communication energy consumed by user $k$ when transmitting data to the UAV at time $n$ is given as 
\begin{align}
    E_k^{cm}[n]=t^{cm}p_k[n],\forall n\in \mathcal{N},k\in\mathcal{K},
\end{align}
and the computation energy of user $k$ at time $n$ is
\begin{align}
    E_k^{cp}[n]=a_k[n]D_k I\psi C f_{cpu,k}^2,\forall n\in\mathcal{N},k\in\mathcal{K},
\end{align}
where $\psi$ is the chip coefficient \cite{Yang2021b}, and other related parameters are defined in (\ref{time_cp_user}).

The FL server (i.e., UAV) consumes energy for model aggregation, model broadcasting, and flight. Since the aggregation energy depends on the UAV's CPU frequency and this work focuses on the effect of UE data size on the FL, we omit model aggregation energy from consideration. The energy consumption for model broadcasting is given by
\begin{align}
    E^{bc}[n]=t^{bc}p_{UAV}[n],\forall n\in \mathcal{N}\setminus\{N-1\}.
\end{align}

Let $P^{fly}(v_{UAV})$ represent the power consumption of a rotary-wing UAV flying at speed $v_{UAV}$ {\cite{Zeng2019}}.
The UAV's flight-related energy consumption includes two operations: flying and hovering, given by 
\begin{align}
    &E^{fly}[n]=t^{fly}[n]P^{fly}(v_{UAV}), \forall n\in\mathcal{N},\\
    &E^{hov}[n]=t^{hov}[n]P^{fly}(0), \forall n\in\mathcal{N}.
\end{align}
 In summary, the total energy consumption is given by
\begin{align}
    &E^{tot}=\sum_{n=1}^{N}E^{fly}[n]+\sum_{n=1}^{N-1}\Biggl\{\sum_{k=1}^K E_k^{cp}[n]+E_k^{cm}[n]\nonumber\\
    &\;\;\;\;\;\;\;\;\;\;\;\;\;+E^{hov}[n]\Biggr\}+\sum_{n=1}^{N-2}E^{bc}[n].
\end{align}

\subsection{Problem Formulation}
To minimize the total energy consumption while ensuring a target FL model accuracy, we propose a joint design problem involving the UAV trajectory $\textbf{q}=\{\textbf{q}[n],\forall n \in\mathcal{N}\cup\{0,N\}\}$, UE FL participation $\textbf{a}=\{a_k[n],\forall n\in\mathcal{N},k\in\mathcal{K}\}$, UE transmit power $\textbf{p}_{UE}=\{p_k[n],\forall n\in\mathcal{N},k\in\mathcal{K}\}$, UAV transmit power $\textbf{p}_{UAV}=\{p_{UAV}[n],\forall n\in \mathcal{N}\setminus\{N-1\}\}$, UE local data size $\textbf{D}=\{D_k,\forall k\in\mathcal{K}\}$, and UAV hovering time $\textbf{t}^{hov}=\{t^{hov}[n],\forall n\in\mathcal{N}\}$. The joint design problem is
\begin{align}
\mathbf{(P1)}\;&\min_{\{\textbf q,\textbf{a},\textbf{p}_{UE},\textbf{p}_{UAV},\textbf{D},\textbf{t}^{hov}\}}E^{tot}\nonumber\\
s.t.\;& (\ref{eq_uav_ini}), (\ref{eq_uav_fin}), (\ref{cons_communacation_hovering}), (\ref{cons_time_UE_less_UAV}), (\ref{cons_total_time}), (\ref{UE_power}),\nonumber\\
 &(\ref{UE_tran_cons}), (\ref{UAV_power}), (\ref{UAV_broad_cons}), (\ref{model_accuracy_cons}), (\ref{tot_data_size}), (\ref{tot_UE_par}),\nonumber
\end{align}

The joint design problem (\(\mathbf{P1}\)) is non-convex and involves integer programming due to the FL participation variables. Given the complexity of jointly optimizing all variables, we address this problem in two phases.

\section{Two-phase Convex Optimization Design}\label{2-phase}

In (\(\mathbf{P1}\)), the UE FL participation variable $a_k[n]$ and the UE local data size $D_k$ can be combined into a single new variable, defined as $\mathbf{\tilde D}=\{D_k[n]=a_k[n]D_k,\forall n\in\mathcal{N},k\in\mathcal{K}\}$, where 
\begin{align}
    D_k[n]\in\{0,D_k\}.\label{Data_asso_integrated}
\end{align}

By substituting $(\ref{Data_asso_integrated})$, we replace the participation and data size variables in all constraints of (\(\mathbf{P1}\)). However, the constraints (\ref{UE_tran_cons}) and (\ref{UAV_broad_cons}) depend solely on the participation variables, and we introduce the sign function $sgn(D_k[n])\in\{0,1\}$ to indicate UE participation in FL. Since the sign function is non-convex, we employ an approximation {\cite{Sadeghi2019}}:
\begin{align}
    a_k[n]=sgn(D_k[n])&\approx\frac{e^{2\beta D_k[n]}-1}{e^{2\beta D_k[n]}+1} \nonumber \\ 
    &\triangleq \tilde a_k[n],\forall n\in\mathcal{N},k\in\mathcal{K},\label{sgn_tilde_a}
\end{align}
where $\beta >0$ controls the approximation accuracy—the larger the $\beta$, the closer the function approximates the true sign function. Here we set $\beta=5$. Moreover, $\tilde{a}_k[n]$ is concave for $D_k[n] >0$. Hence, the joint design problem ($\textbf{P1}$) can be equivalently rewritten as
\begin{align}
&\qquad\qquad\qquad\mathbf{(P2)}\;\min_{\{\textbf q,\textbf{p}_{UE},\textbf{p}_{UAV},{\textbf{D}},\mathbf{\tilde D},\textbf{t}^{hov}\}} {E^{tot}}\nonumber\\
s.t.\;& (\ref{eq_uav_ini}), (\ref{eq_uav_fin}),(\ref{cons_communacation_hovering}) , (\ref{cons_total_time}), (\ref{UE_power}), (\ref{UAV_power}), {(\ref{Data_asso_integrated})},\nonumber\\
&t^{cm}+D_k[n]\Phi_k \leq t^{fly}[n]+t^{hov}[n], \forall k\in \mathcal{K},\forall n\in \mathcal{N}, \label{Phase1_timeslot}\\
&\tilde a_k[n]\mathcal{Q}\leq t^{cm}R_k[n],\forall k \in \mathcal{K},\forall n\in\mathcal{N}, \label{Phase1_com_cons}\\
&\tilde a_k[n+1]\mathcal{Q}\leq t^{bc}R_k^{bc}[n],\forall k \in \mathcal{K},\forall n\in \mathcal{N}\setminus\{N-1\},\label{Phase1_broadcast_cons} \\ 
&\sum_{n=1}^{N-1} D_k[n]\geq D_{th},\forall k \in \mathcal{K},\label{Phase1_DataSize}\\
&\frac{L}{2}\Biggl[\omega^{(N-1)I}\mathbb{E}\left[\|\bar{\mathbf{w}}^0-\mathbf{w}^*\|^2\right]+A_1\left(\frac{1-\omega^{(N-1)I}}{\eta\mu}\right)\nonumber\\
    &+\eta^2\sum_{n=1}^{N-1}\left(\omega^{((N-1)-n)I}\left(\frac{1-\omega^I}{\eta\mu}\right)\sum_{k=1}^K(\tilde D_k[n])^2\epsilon^2_v\right)\Biggr]\nonumber\\
    &\leq\epsilon_G, \label{Phase1_acc}\\
    &\sum_{k=1}^{K} \tilde{a}_k[n]\geq a_{min},\forall n \in \mathcal{N}, \label{Phase1_UE_par}
\end{align}
where {$E^{cp}_k= D_k[n]I\psi C f_{cpu,k}^2$} in $E^{tot}$, and $\tilde D_k[n]={D_k[n]}{(\sum^K_{j=1}D_j[n])^{-1}}$. 

\vphantom{\hphantom{Although we combine the UE FL participation variable $a_k[n]$ with the UE local data size variable $D_k$ to produce a new variable $D_k[n]$, the joint design problem ($\mathbf{P2}$) is still non-convex, so we will relax the constraints and perform a two-phase optimization.}}

\subsection{Phase I: Optimization with Relaxed Data Size Restrictions}
{In Phase I, we firstly relax the restriction on the binary discrete size of the data (\ref{Data_asso_integrated}), yielding
\begin{align}
    D_k[n]\geq0,\forall n\in\mathcal{N},k\in\mathcal{K}. \label{Phase1_relaxDk}
\end{align}
Hence, the relaxed problem of ($\textbf{P2}$) becomes
\begin{align}
\mathbf{(P3)}\;&\min_{\{\textbf q,\textbf{p}_{UE},\textbf{p}_{UAV},\mathbf{\tilde D},\textbf{t}^{hov}\}}E^{tot}\nonumber\\
s.t.\;& (\ref{eq_uav_ini}), (\ref{eq_uav_fin}), (\ref{cons_communacation_hovering}),(\ref{cons_total_time}), (\ref{UE_power}), (\ref{UAV_power}),\nonumber\\
&(\ref{Phase1_timeslot}), (\ref{Phase1_com_cons}), (\ref{Phase1_broadcast_cons}), (\ref{Phase1_DataSize}), (\ref{Phase1_acc}), (\ref{Phase1_UE_par}), (\ref{Phase1_relaxDk}),\nonumber
\end{align}
where the data size $D_k$ is no longer an optimization variable due to the relaxation of (\ref{Data_asso_integrated}) in (\ref{Phase1_relaxDk}). The problem ($\mathbf{P3}$) remains non-convex due to the constraints (\ref{Phase1_timeslot})-(\ref{Phase1_broadcast_cons}) and (\ref{Phase1_acc}). We next convert these constraints into convex ones.}

We first introduce an auxiliary variable $d^{lb}[n]$ which satisfies
\begin{align}
    d^{lb}[n]\leq\|\textbf{q}[n]-\textbf{q}[n-1]\|,\forall n\in\mathcal{N}. \label{d_lb}
\end{align}
By using (\ref{d_lb}), the constraint (\ref{Phase1_timeslot}) can be replaced with a lower bound constraint:
\begin{align}
    t^{cm}+D_k[n]\Phi_k\leq \frac{d^{lb}[n]}{v_{UAV}}+t^{hov}[n].\label{timeslot_rewrite}
\end{align}
To address the non-convex constraint (\ref{d_lb}), we apply a first-order Taylor expansion. Since $\|\textbf{q}[n]-\textbf{q}[n-1]\|^2$ is convex in both $\textbf{q}[n]$ and $\textbf{q}[n-1]$, we replace the constraint (\ref{d_lb}) by applying its first-order Taylor expansion at a given point $\textbf{q}^r[n]$, yielding a convex lower bound constraint: 
\begin{align}
    &(d^{lb}[n])^2\leq\|\textbf{q}^r[n]-\textbf{q}^r[n-1]\|^2\nonumber\\
    &\qquad\qquad+2(\textbf{q}^r[n]-\textbf{q}^r[n-1])^T(\textbf{q}[n]-\textbf{q}^r[n])\nonumber\\
    &\qquad\qquad-2(\textbf{q}^r[n]-\textbf{q}^r[n-1])^T(\textbf{q}[n-1]-\textbf{q}^r[n-1]),\nonumber\\
    &\qquad\qquad\qquad\qquad\qquad\qquad\qquad\quad\forall n\in\mathcal{N},k\in\mathcal{K}.\label{con_1st_order_timeslot}
\end{align}

Since the transmission rate $R_k[n]$ in (\ref{Phase1_com_cons}) is neither convex nor concave in $\textbf{q}[n]$ and $p_k[n]$, we convexify (\ref{Phase1_com_cons}) by finding a concave lower bound for $R_k[n]$. Following {\cite{Fu2024}}, we introduce two auxiliary variables, $A_k[n]$ and $B_k[n]$, given by
\begin{align}
    exp(A_k[n])=\tilde{g}_k[n],\forall n\in\mathcal{N},k\in\mathcal{K};\label{expA}\\
    exp(B_k[n])=p_k[n],\forall n\in\mathcal{N},k\in\mathcal{K}.\label{expB}
\end{align}
Then we express $R_k[n]$ as a difference of two logarithmic functions as $R_k[n]=\frac{W}{ln(2)}(R_1[n]+R_{2,k}[n])$, where we define $R_1[n]\triangleq\ln\left( \sum_{i=1}^{K}e^{B_i[n]+A_i[n]}+\sigma_z^2 \right)$ and $R_{2,k}[n]\triangleq-\ln\left( \sum_{i=1,i\neq k}^{K}e^{B_i[n]+A_i[n]}+\sigma_z^2 \right)$. In the following, we separately derive concave lower bounds for $R_1[n]$ and $R_{2,k}[n]$ in terms of the trajectory variable $\textbf{q}[n]$. 

Given any ${\textbf q}^r[n]$, $R_1[n]$ is lower bounded by a concave function $R_1^{lb}[n]$ in terms of $\textbf q[n]$ and $B_k[n]$:
\begin{align}
    &R_1[n]\geq R_{1}[n]\biggl|_{\textbf q[n]=\textbf q^r[n],B_k[n]=B_k^r[n]}\nonumber\\
    &\;+\sum_{i=1}^K \frac{e^{B_i^r[n]+A_i^r[n]}}{\sum_{j=1}^K e^{B_j^r[n]+A_j^r[n]}+\sigma^2_z}\left(A^{lb}_i[n]-A_i^r[n]\right)\nonumber\\
    &\;+\sum_{i=1}^K \frac{e^{B_i^r[n]+A_i^r[n]}}{\sum_{j=1}^K e^{B_j^r[n]+A_j^r[n]}+\sigma^2_z}\left(B_i[n]-B_i^r[n]\right)\nonumber\\
    &\;\;\;\;\triangleq R_1^{lb}[n],\forall n\in\mathcal{N},k\in\mathcal{K}\label{R_1_llb},
\end{align}
where $A_k^r[n]$ and $B_k^r[n]$ are calculated from $\textbf{q}[n]$ and $B_k[n]$  with $\textbf q[n]=\textbf q^r[n], B_k[n]=B_k^r[n]$, and other terms are given by
    \begin{align}
    &A^{lb}_i[n]=\ln\left(\frac{\left(c(4\pi f_c)^{-1}\right)^2}{S_i^r[n]}\right)\nonumber\\
    &\qquad\quad-\frac{\left(\|\textbf{q}[n]-\textbf{g}_{UE,i}\|^2+H^2-S_i^r[n]\right)}{S_i^r[n]};\label{Alb_def}\\
    &S_i^r[n]=\|\textbf{q}^r[n]-\textbf{g}_{UE,i}\|^2+H^2.\label{S_i_r}
\end{align}

Next, from (\ref{eq_pathloss}) and (\ref{eq_SINR}), $\tilde{g}_k[n]$ is non-convex in $\textbf{q}[n]$. An auxiliary variable $\tilde{A}_k[n]$ is introduced to {ensure $e^{A_k[n]}=\tilde{g}_k[n] \leq e^{\tilde A_k[n]}$}, yielding 
\begin{align}
        \frac{\|\textbf{q}[n]-\textbf{g}_{UE,k}\|^2+H^2}{\left(c(4\pi f_c)^{-1}\right)^2}\geq e^{-\tilde A_k[n]}.\label{aux_A_rewrite1}
\end{align}
Hence, the rate formula $R_{2,k}[n]$ is lower bounded by a concave function $R_{2,k}^{lb}[n]$:
\begin{align}
    R_{2,k}[n]&\geq-\ln\left(\sum_{i=1,i\neq k}^{K}e^{B_i[n]+\tilde A_i[n]}+\sigma_z^2\right)\nonumber\\
              &\triangleq R_{2,k}^{lb}[n],\forall n\in\mathcal{N},k\in\mathcal{K}.
\end{align}}The imposed constraint (\ref{aux_A_rewrite1}) is, however, non-convex, and we apply a first-order Taylor expansion for the left-hand-side quadratic function at a given point $\textbf{q}[n]=\textbf{q}^r[n]$, leading to a convex lower bound constraint:
\begin{align}
    &\frac{\|\textbf{q}^r[n]-\textbf{g}_{UE,k}\|^2+2(\textbf{q}^r[n]-\textbf{g}_{UE,k})^T(\textbf{q}[n]-\textbf{q}^r[n])+H^2}{\left(c(4\pi f_c)^{-1}\right)^2}\nonumber\\
    &\qquad\qquad\qquad\geq e^{-\tilde A_k[n]}, \forall n\in\mathcal{N},k\in\mathcal{K}. \label{aux_A_lb}
\end{align}

By replacing $R_1[n]$ and $R_{2,k}[n]$ in $R_k[n]$ with the derived concave lower bounds $R_1^{lb}[n]$ and $R_{2,k}^{lb}[n]$, the constraint (\ref{Phase1_com_cons}) can be rewritten as
\begin{align}
    \tilde a_k[n]\mathcal{Q}\leq \frac{t^{cm}W}{\ln(2)}\left(R_1^{lb}[n]+R_{2,k}^{lb}[n]\right),\forall k \in \mathcal{K},\forall n\in\mathcal{N}.\label{replace_r1r2}
\end{align}
From (\ref{sgn_tilde_a}), $\tilde a_k[n]$ is concave for $D_k[n]\geq0$, rendering constraint (\ref{replace_r1r2}) non-convex. To convexify it, we apply a first-order Taylor expansion at a given point $D_k^r[n]$, resulting in
\begin{align}
    \tilde a_k[n]&\leq\frac{e^{(2\beta D_k^r[n])}-1}{e^{(2\beta D_k^r[n])}+1}+\frac{4\beta e^{(2\beta D_k^r[n])}}{\left(e^{(2\beta D_k^r[n])}+1 \right)^2}(D_k[n]-D_k^r[n])\nonumber\\
            &\triangleq \check a_k[n],\forall n\in\mathcal{N},k\in\mathcal{K}.\label{check_a}
\end{align}
By utilizing (\ref{check_a}), the contraint (\ref{replace_r1r2}) can be convexified as
\begin{align}
    \check a_k[n]\tilde{\mathcal{Q}}\leq R_1^{lb}[n]+R_{2,k}^{lb}[n],\forall k \in \mathcal{K},\forall n\in\mathcal{N},\label{com_approx_con}
\end{align}
where $\tilde{\mathcal{Q}}=\frac{\mathcal{Q}\ln(2)}{t^{cm}W}$.

Next, we deal with the non-convex constraint (\ref{Phase1_broadcast_cons}). By defining a variable $C[n]$, where $exp(C[n])=p_{UAV}[n]$, for $n\in \mathcal{N}\setminus\{N-1\}$, the constraint (\ref{Phase1_broadcast_cons}) can be convexified by
    \begin{align}
        &\ln\left(e^{(\tilde a_k^r[n+1]\hat{\mathcal{Q}})}-1\right)+\frac{\hat{\mathcal{Q}}e^{\hat{\mathcal{Q}}\tilde a_k^r[n+1]}}{e^{\hat{\mathcal{Q}}\tilde a_k^r[n+1]}-1}\left(\check{a}_k[n+1]-\tilde a_k^r[n+1]\right)\nonumber\\
        &\leq C[n]+A^{lb}_k[n]-\ln\left(\sigma_z^2\right),\forall n\in\mathcal{N}\setminus\{N-1\},k\in\mathcal{K},\label{bro_approx_con}
    \end{align}
    where $\hat{\mathcal{Q}}=\frac{\mathcal{Q}\ln(2)}{t^{bc}W}$, $\tilde a_k^r[n]=\frac{e^{(2\beta D_k^r[n])}-1}{e^{(2\beta D_k^r[n])}+1}$.

In the model accuracy constraint (\ref{Phase1_acc}), $\tilde D_k[n]$ is non-convex in $D_k[n]$ for $k\in\mathcal{K}$. To address this, we introduce an auxiliary variable $\hat D[n]$ satisfying the constraint:
\begin{align}
    \hat D[n]\leq\left(\sum_{j=1}^K D_j[n]\right)^2,\forall n \in \mathcal{N}.\label{aux_D}
\end{align}
Applying (\ref{aux_D}), we replace the original model accuracy constraint (\ref{Phase1_acc}) with the following upper bound:
\begin{align}
    &\frac{L}{2}\Biggl[\omega^{(N-1)I}\mathbb{E}\|\bar{\mathbf{w}}^0-\mathbf{w}^*\|^2+A_1\left(\frac{1-\omega^{(N-1)I}}{\eta\mu}\right) + \label{Phase1_acc_up}\\
    &\eta^2\sum_{n=1}^{N-1}\left(\omega^{((N-1)-n)I}\left(\frac{1-\omega^I}{\eta\mu}\right)\sum_{k=1}^K\frac{(D_k[n])^2}{\hat{D}[n]}\epsilon^2_v\right)\Biggr]\leq\epsilon_G. \nonumber 
\end{align}
The constraint (\ref{Phase1_acc_up}) is now convex, since the function $\frac{(D_k[n])^2}{\hat{D}[n]}$ is convex in $D_k[n]$ and $\hat{D}[n]>0$. Next, we deal with the non-convex constraints (\ref{aux_D}). Noting that $\left(\sum_{j=1}^K D_j[n]\right)^2$ is convex in $D_j[n]$, the constraint (\ref{aux_D}) is then convexified by using its first-order Taylor expansion at a given point $D_j[n]=D_j^r[n]$:
\begin{align}
    &\hat D[n]\leq \left(\sum_{j=1}^K D_j^r[n]\right)^2 \label{aux_D_lb} \\
    &+\sum_{j'=1}^K\left(2\left(\sum_{j''=1}^K D_{j''}^r[n]\right)\left(D_{j'}[n]-D_{j'}^r[n]\right)\right), \forall n\in\mathcal{N}. \nonumber
\end{align}

With the transformed convex constraints and introduced auxiliary variables, the problem $\mathbf{(P3)}$ can be transformed as
\begin{align}
&\mathbf{(P4)}\;\min_{\{\textbf q,\tilde{\textbf{A}},\textbf{B},\textbf{C},\mathbf{\tilde D},\mathbf{\hat D},\textbf{t}^{hov}\}} {E^{tot}}\nonumber\\
s.t.\;& (\ref{eq_uav_ini}), (\ref{eq_uav_fin}), (\ref{cons_communacation_hovering}),(\ref{cons_total_time}), (\ref{UE_power}), (\ref{UAV_power}), (\ref{Phase1_DataSize}), (\ref{Phase1_UE_par}),\nonumber\\
&(\ref{Phase1_relaxDk}), (\ref{timeslot_rewrite}), (\ref{con_1st_order_timeslot}), (\ref{aux_A_lb}), (\ref{com_approx_con}), (\ref{bro_approx_con}), (\ref{Phase1_acc_up}), (\ref{aux_D_lb}),\nonumber
\end{align}
where $\tilde{\textbf{A}}=\{\tilde A_k[n],\forall n\in\mathcal{N},k\in\mathcal{K} \}$, $\textbf{B}=\{B_k[n],\forall n\in\mathcal{N},k\in\mathcal{K}\}$, $\textbf{C}=\{C[n],\forall n\in \mathcal{N}\setminus\{N-1\}\}$, $\hat{\textbf{D}}=\{\hat D[n],\forall n\in\mathcal{N}\}$. With relaxed data size, the joint design problem ($\mathbf{P3}$) can be solved via ($\mathbf{P4}$) for given $\textbf{q}^r[n]$, $B^r_k[n]$ and $D^r_k[n]$. Utilizing the SCA method {\cite{Razaviyayn2014}} and the CVX optimization tool {\cite{Grant2014}}, we iteratively solve ($\mathbf{P4}$) to determine the UAV trajectory, UE and UAV transmit power, relaxed data size, and UAV hovering time.

\subsection{Phase II: Re-optimization with Data Size Restrictions}
The data size relaxation in Phase I allows UE data sizes to vary across time slots. In Phase II, We refine the solution by considering the data size restriction. As $D_k[n]$ in Phase I encapsulates the UE FL participation $a_k[n]$ and UE local data $D_k$, it reflects not only the UE FL participation but also the outcomes of joint optimization over several variables (e.g., UAV trajectory, transmit power). In Phase II, we use $D_k[n]$ to infer the UE participation and then re-optimize $D_k$ accordingly. From (\ref{sgn_tilde_a}), ${a}_k[n]$ is computed and quantized as:
\begin{align}
a_k[n]=\left\{
\begin{aligned}
     &1,\;\;\tilde a_k[n]\geq 0.5 \;;\\
     &0,\;\; otherwise.\\
\end{aligned}
\right.\label{ak_redefine}
\end{align}

{With ($\mathbf{P4}$), we optimize the data size $\textbf{D}$ under fixed $a_k[n]$ using (\ref{ak_redefine}) while also re-optimizing $\textbf{q}$, $\textbf{p}_{UE}$, $\textbf{P}_{UAV}$, and $\textbf{t}^{hov}$. To this end, we first replace all constraints involving $D_k[n]$ from Phase I with $a_k[n]D_k$, except for (\ref{com_approx_con}) and (\ref{bro_approx_con}). Instead, these two constraints are traced back to (\ref{Phase1_com_cons}) and (\ref{Phase1_broadcast_cons}). By replacing $\tilde a_k[n]$ directly with $a_k[n]$ in (\ref{Phase1_com_cons}) and (\ref{Phase1_broadcast_cons}), they can be convexified using the same approach as in (\ref{replace_r1r2}) and (\ref{bro_approx_con}), resulting in:
\begin{align}
    &a_k[n]\mathcal{Q}\leq \frac{t^{cm}W}{\ln(2)}\left(R_1^{lb}[n]+R_{2,k}^{lb}[n]\right),\forall k \in \mathcal{K},\forall n\in\mathcal{N};\label{phase2_UAV_com_cons_re}\\
    &\ln\left(e^{(a_k[n+1]\hat{\mathcal{Q}})}-1\right)\leq C[n]+A^{lb}_k[n]-\ln\left(\sigma_z^2\right),\nonumber\\
    &\qquad\qquad\qquad\qquad\qquad\forall n\in\mathcal{N}\setminus\{N-1\},k\in\mathcal{K}.\label{phase2_UAV_broad_cons_re}
\end{align}
By leveraging the SCA method {\cite{Razaviyayn2014}} and CVX tool {\cite{Grant2014}}, the UAV trajectory, UE and UAV transmit power, UAV hovering time, and data size can be jointly re-optimized for fixed points of $\textbf{q}^r[n]$, $B_k^r[n]$, and $D_k^r$.}

\section{Numerical Simulation}\label{Simulation}
\subsection{Simulation Settings}

We consider a UAV flying over a $600$ m $\times$ $600$ m area at a fixed altitude of $150$ m. The UAV is initialized at two possible positions: $[0,300,150]$ m and $[600,300,150]$ m. The maximum UAV velocity is $10$ m/sec. The number of UEs is set to $K = 6$, with coordinates given by $\textbf g_{UE,1}=[0,400]^T$, $\textbf g_{UE,2}=[100,600]^T$, $\textbf g_{UE,3}=[100,400]^T$, $\textbf g_{UE,4}=[400,600]^T$, $\textbf g_{UE,5}=[400,400]^T$, $\textbf g_{UE,6}=[500,400]^T$. The mission duration is set to $T = 500$ s, divided into $50$ time slots. The transmission time, aggregation time, and broadcast time are respectively set to $t_{cm}=2$ sec,  $t_{agg}=0.5$ sec, and $t_{bc}=0.5$ sec. The maximum transmit power for the UAV and UEs are $p_{uav}^{max}=30$ dBm and $p_{UE}^{max}=31.8$ dBm. The UAV’s onboard CPU operates at $2$ GHz with a chip coefficient of $\zeta=10^{-25}$ and a computational resource requirement of $C=10$ for computing one bit. The system operates at a $2.4$ GHz carrier frequency and a $20$ MHz bandwidth. The noise power is $\sigma^2_z$ is $-80$ dBm. The parameters of the FL loss function and UAV's flight-related energy consumption are set to be the same as in [xx]. Additionally, the minimum number of participating UEs is $a_{min}=2$, and the minimum total data requirement for FL is $D_{th}=50$ Mb. {The FL model accuracy threshold is $\epsilon_{G} = 10$.} Unless otherwise stated, the above values serve as the default settings.

\subsection{Simulation Results}

\begin{figure}[htbp]
\centering
\includegraphics[width=0.45\textwidth]{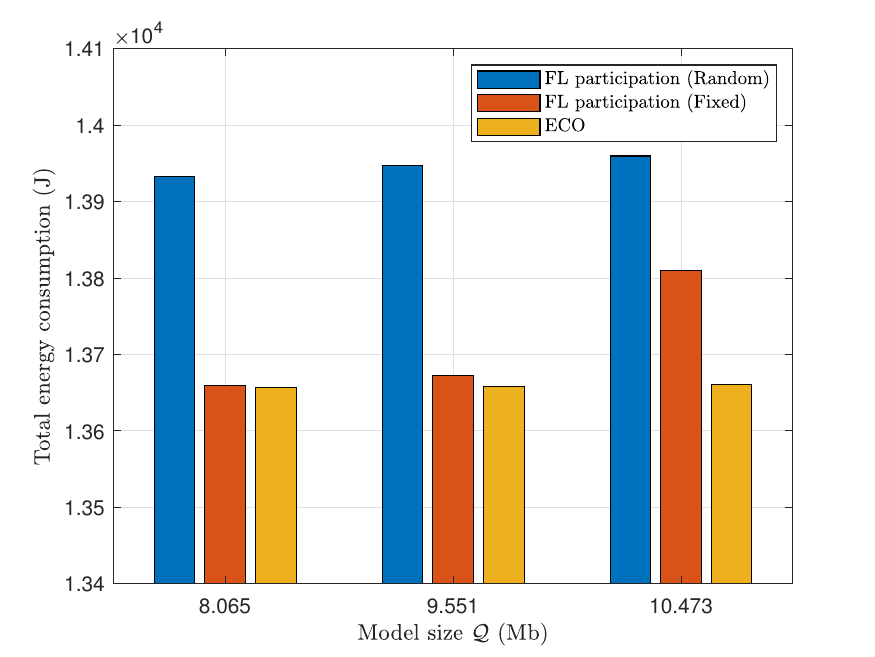}
\caption{Performance of different FL participation design methods.}
\label{fig_baseline_compare}
\end{figure}

{
To access the impact of FL participation, Fig.  \ref{fig_baseline_compare} compares the following baseline methods: (1) random participation {\cite{Zeng2020}}, where two UEs are randomly selected per time slot; and (2) fixed participation {\cite{Yang2021b}}, where all UEs participate in every time slot. For the ECO method, the initial solution assumes full participation by all UEs. The results show that when the model size is small ($\mathcal{Q} = 8.065$ Mb), the performance gap between the fixed method and ECO is small. As the model size increases, the ECO method outperforms both baselines in terms of energy efficiency, highlighting the growing benefit of optimizing FL participation under larger model sizes.

}

\begin{figure}[htbp]
\centering
\includegraphics[width=0.45\textwidth]{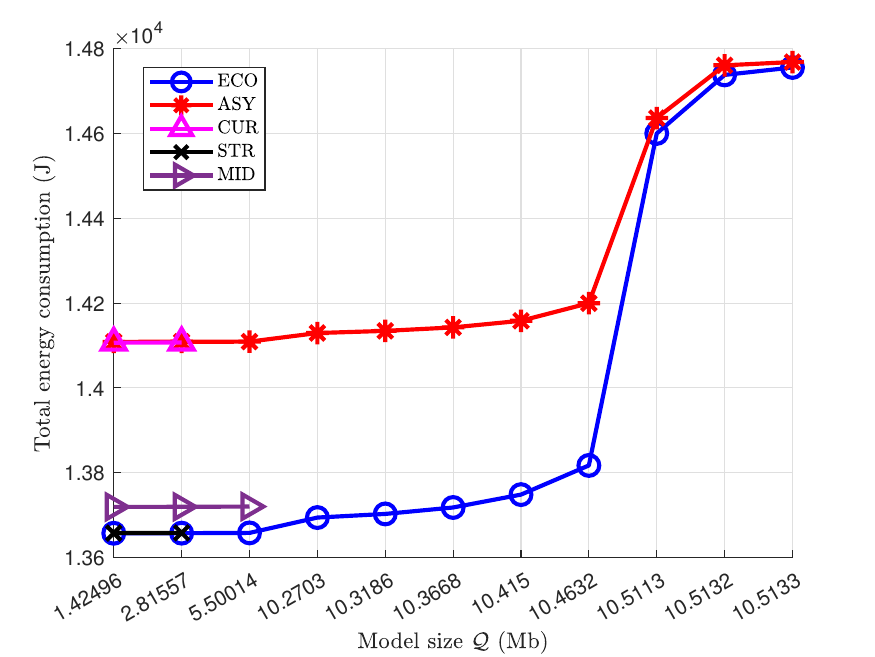}
\caption{Energy consumption of different fixed UAV trajectory designs and proposed ECO method for various FL model sizes.}
\label{fig_diff_method}
\end{figure}

To evaluate the impact of UAV trajectory design, Fig. \ref{fig_diff_method} compares the proposed ECO method with four heuristic trajectories: (1) Curve (CUR), (2) Straight (STR), (3) Middle (MID), and (4) Asymptotic (ASY), as illustrated in Fig. \ref{fig_traj}. For these methods, the UAV trajectory is fixed, while the remaining variables (UE FL participation, UAV and UE transmit power, and UE local data size, UAV hovering time) are jointly optimized using the SCA. 

The results show that the proposed ECO method consistently achieves the lowest energy consumption across various FL model sizes. Although the STR trajectory performs comparably to the ECO method for small model sizes (e.g., $\mathcal{Q} = 1.42496$ Mb and $2.81557$ Mb), it becomes infeasible for larger model sizes (e.g., $\mathcal{Q}$ ranges between $5.50014$ Mb and $10.5133$ Mb) due to constraint violations.

\begin{figure}[htbp]
\centering
\includegraphics[width=0.45\textwidth]{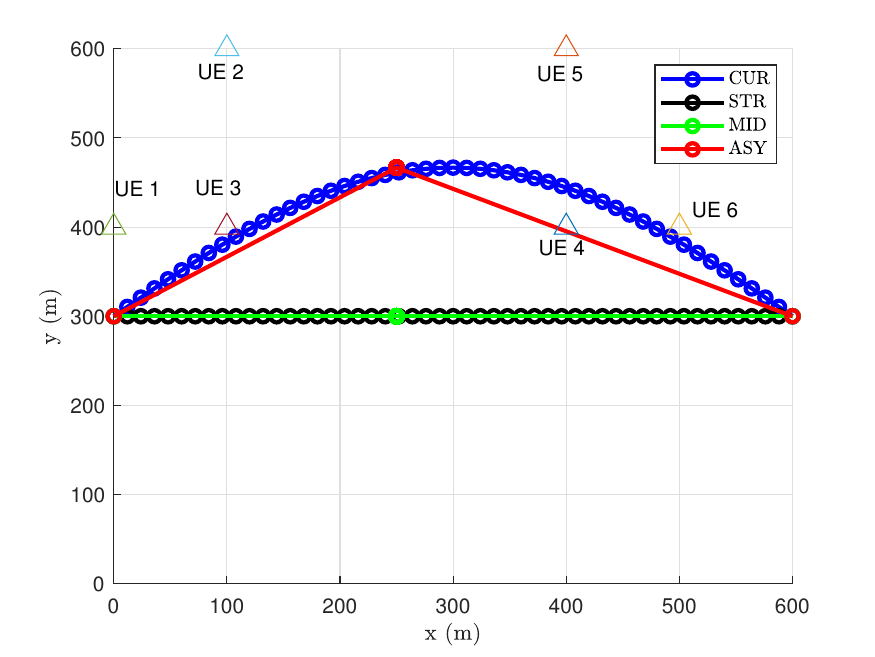}
\caption{The UAV trajectories of different heuristic methods.}
\label{fig_traj}
\end{figure}

\section{Conclusion}\label{Conclusion}
{This paper proposed a UAV-assisted FL framework to minimize total system energy consumption through the joint design of UAV trajectory, user participation, power control, and data volume allocation. A convergence analysis of the FL model with multiple local updates was conducted to examine the impact of user participation and data volume control on model learning accuracy, thereby providing a theoretical foundation for energy-efficient resource management. To tackle the resulting non-convex problem, we employed an SCA approach. Simulation results verified that the proposed ECO methods significantly outperform compared methods in terms of energy consumption and model convergence.}
\appendices
\section{Proof of Theorem \ref{model_acc_Theorem1}}\label{Theorem1_proof}
\setcounter{equation}{0}
\renewcommand\theequation{A.\arabic{equation}}
{
Using the $L$-smooth of $f_G(\mathbf{w})$ from Remark \ref{assumption_lsmooth_remark_2}, we have
    \begin{align}\label{L-smooth_simplified}
        \mathbb{E}\big[f_G(\bar{\mathbf{w}}^{i+1})-f_G(\mathbf{w}^*)\big]\leq \frac{L}{2} \mathbb{E}\big[\|\bar{\mathbf{w}}^{i+1}-\mathbf{w}^*\|^2 \big].
    \end{align}An upper bound for $\frac{L}{2}\mathbb{E}\big[\|\bar{\mathbf{w}}^{i+1}-\mathbf{w}^*\|^2$\big] is then provided.
}
\begin{theorem}\label{modeliter_modelopt_upper_bound}
 Under Assumptions 1-4 and a learning rate $\eta\leq \frac{1}{2L}$, $\mathbb{E}\big[\|\bar{\mathbf{w}}^{i+1}-\mathbf{w}^*\|^2 \big]$ can be upper bounded by
\begin{align}
    &\mathbb{E}\big[\|\bar{\mathbf{w}}^{i+1}-\mathbf{w}^*\|^2 \big] \leq\omega^{i+1}\mathbb{E}\big[\|\bar{\mathbf{w}}^0-\mathbf{w}^*\|^2\big]+A_1\left(\frac{1-\omega^{i+1}}{\eta\mu}\right)\nonumber\\
    &\;\;\;\;\;\;\;\;\;\;\;\;\;\;\;\;\;\;+\eta^2\sum_{l=0}^{i}\left(\omega^{i-l}\sum_{k=1}^K(\bar D_k^{l})^2\epsilon^2_v\right),\label{accuracy_proof}
\end{align}
where $\omega=1-\eta\mu$, $A_1=(1+\frac{\zeta}{2\eta})I^2\eta^2\epsilon_s^2+\frac{\eta L^2\epsilon_w}{2}(\zeta+4\eta)$, $\zeta=2\eta(1-\eta2L)$, and $\bar D_k^i=\frac{a_k^i D_k}{\sum_{j=1}^Ka_j^i D_j}$.
\end{theorem}
\begin{proof}
See Appendix \ref{modeliter_modelopt_upper_bound_proof} for the detailed proof.
\end{proof}

By substituting Theorem \ref{modeliter_modelopt_upper_bound} into (\ref{L-smooth_simplified}), an upper bound on the FL model accuracy is obtained as expressed in (\ref{Theo1_acc}).

{}

\section{Proof of Theorem \ref{modeliter_modelopt_upper_bound}}\label{modeliter_modelopt_upper_bound_proof}
From {(\ref{virtual_update})}, the expected difference between the weighted model $\bar{\textbf{w}}^{i+1}$ and the optimal FL model $\textbf{w}^*$ is given as
\setcounter{equation}{0}
\renewcommand\theequation{B.\arabic{equation}}
\begin{align}\label{weight_optimal_expression}
    &\mathbb{E}\left[\|\bar{\mathbf{w}}^{i+1}-\mathbf{w}^*\|^2\right]
    \nonumber\\
    &=\mathbb{E}\left[\left\|\left(\bar{\mathbf{w}}^i-\eta\sum_{k=1}^K\bar D_k^i\nabla f_{L,k}(\mathbf{w}_k^i,s_k^i)\right)-\mathbf{w}^*\right\|^2\right]\nonumber\\
    &=\mathbb{E}\Biggl[\Biggl\| \underbrace{\bar{\mathbf{w}}^i-\mathbf{w}^*- \sum_{k=1}^K\bar D_k^i\left(\eta\nabla f_{L,k}(\mathbf{w}_k^i)\right)   }_{\triangleq A_2}\nonumber\\
    &\;\;+\underbrace{\sum_{k=1}^K\bar D_k^i\left(\eta\nabla f_{L,k}(\mathbf{w}_k^i)\right)-     \sum_{k=1}^K\bar D_k^i\left(\eta\nabla f_{L,k}(\mathbf{w}_k^i,s_k^i)\right)  }_{\triangleq A_3}\Biggr\|^2\Biggr]\nonumber\\     &=\mathbb{E}\left[\left\|A_2\right\|^2\right]+\mathbb{E}\left[\left\|A_3\right\|^2\right],
    \end{align}
    where $\mathbb{E}[A_2A_3]=0$ since $\mathbb{E}\left[A_3\right]=0$.
    {}By expanding $\|A_2\|^2$, we can get
\begin{align}\label{A_2}
    \|A_2\|^2=&\left\|\bar{\mathbf{w}}^i-\mathbf{w}^*\right\|^2\underbrace{-2\eta\sum_{k=1}^K\bar D_k^i\left\langle \bar{\mathbf{w}}^i-\mathbf{w}^*,\nabla f_{L,k}(\mathbf{w}_k^i)\right\rangle}_{\triangleq A_4}\nonumber\\
    &+\underbrace{\left\|\sum_{k=1}^K\bar D_k^i\left(\eta\nabla f_{L,k}(\mathbf{w}_k^i)\right)\right\|^2}_{\triangleq A_5},
\end{align}
where $A_4$ can be rewritten as
\begin{align}\label{A_4}
    A_4&=\underbrace{-2\eta\sum_{k=1}^K\bar D_k^i\left\langle \bar{\mathbf{w}}^i-\mathbf{w}_k^i,\nabla f_{L,k}(\mathbf{w}_k^i)\right\rangle}_{B_1}\nonumber\\
    &\;\;\underbrace{-2\eta\sum_{k=1}^K\bar D_k^i\left\langle\mathbf{w}_k^i-\mathbf{w}^*,\nabla f_{L,k}(\mathbf{w}_k^i)\right\rangle}_{B_2}.
\end{align}
{We then introduce the following lemma.
\begin{lemma}\label{AMGM_CS}
    For any $\mathbf{a}$, $\mathbf{b}$ and $\eta>0$, we have the inequality
    \begin{align}
        -2\langle \mathbf{a},\mathbf{b} \rangle\leq\frac{1}{\eta}\|\mathbf{a}\|^2+\eta\|\mathbf{b}\|^2
    \end{align}
\end{lemma}
\begin{proof}
Details can be found in \cite{Zhou2022} by using Cauchy–Schwarz and arithmetic–geometric mean inequalities.
    \end{proof}}From Lemma \ref{AMGM_CS} and $L$-smoothness, $B_1$ is upper bounded by
\begin{align}\label{B_1_up}
    &B_1\leq\eta\sum_{k=1}^K\bar D_k^i\Biggl(\frac{1}{\eta}\left\|\bar{\mathbf{w}}^i-\mathbf{w}_k^i\right\|^2+\eta\|\nabla f_{L,k}(\mathbf{w}_k^i)\|^2\Biggr)\\
    &\leq\sum_{k=1}^K\bar D_k^i\Biggl(\left\|\bar{\mathbf{w}}^i-\mathbf{w}_k^i\right\|^2+\eta^22L\left(f_{L,k}\left(\mathbf{w}_k^i\right)-f_{L,k}(\mathbf{w}_k^*)\right)\Biggr). \nonumber
\end{align}
Additionally, by successively using {$\mu$-strongly convex and Jensen's inequality}, $B_2$ is upper bounded by
\begin{align}\label{B_2_up}
    B_2&\leq -2\eta\sum_{k=1}^K\bar D_k^i\biggl(f_{L,k}(\mathbf{w}_k^i)-f_{L,k}(\mathbf{w}^*)+\frac{\mu}{2}\|\mathbf{w}_k^i-\mathbf{w}^*\|^2\biggr)\nonumber\\
    &\leq -2\eta\sum_{k=1}^K\bar D_k^i\biggl(f_{L,k}(\mathbf{w}_k^i)-f_{L,k}(\mathbf{w}^*)\biggr)\nonumber\\
    &\;\;\;\;-\eta\mu\left\|{\bar{\mathbf{w}}^{i}}-\mathbf{w}^*\right\|^2.
\end{align}
Using {Jensen's inequality and Remark \ref{assumption_lsmooth_remark_1}}, $A_5$ is bounded by
\begin{align}\label{A_5}
    A_5\leq2L\eta^2\sum_{k=1}^K\bar D_k^i\left(f_{L,k}(\mathbf{w}_k^i)-{f_{L,k}(\mathbf{w}_k^*)}\right) \,.
\end{align}
By combining (\ref{A_4}) and (\ref{B_1_up})--(\ref{A_5}), (\ref{A_2}) is bounded by 
\begin{align}\label{A_2_upp_1}
     \|A_2\|^2&\leq (1-\eta\mu)\left\|\bar{\mathbf{w}}^i-\mathbf{w}^*\right\|^2+\sum^K_{k=1}\bar D_k^i\left\| {\bar{\mathbf{w}}^{i}}-\mathbf{w}_k^i\right\|^2\nonumber\\
    &\;+\underbracea{4L\eta^2\sum^K_{k=1}\bar D_k^i\left(f_{L,k}(\mathbf{w}_k^i)-f_{L,k}(\mathbf{w}_k^*)\right)}\nonumber\\
    &\;\underbracebd{-2\eta\sum^K_{k=1}\bar D_k^i\left(f_{L,k}(\mathbf{w}_k^i)-f_{L,k}(\mathbf{w}^*)\right)}_{B_3\qquad\qquad\qquad\qquad\qquad\qquad\;}.
\end{align}
We then provide an upper bound for $B_3$ as follows.
\begin{lemma}\label{B_3_lemma}
    For given $\eta\leq1/2L$, $B_3$ is upper bounded by
    \begin{align}
        B_3\leq\frac{\eta L^2\epsilon_w}{2}(\zeta+4\eta)+\frac{\zeta}{2\eta}\sum^K_{k=1}\bar D_k^i\left\|\mathbf{w}_k^i-{\bar{\mathbf{w}}^{i}}\right\|^2,
    \end{align}
    where $\zeta=2\eta(1-\eta2L)$.
\end{lemma}
\begin{proof}
    See Appendix \ref{B_3_proof} for the detailed proof.
\end{proof}By using Lemma \ref{B_3_lemma} in (\ref{A_2_upp_1}), $\mathbb{E}\left[\|A_2\|^2\right]$ is bounded by
\begin{align}\label{A_2_upp}
     &\mathbb{E}\left[\|A_2\|^2\right]\leq (1-\eta\mu)\mathbb{E}\left[\left\|\bar{\mathbf{w}}^i-\mathbf{w}^*\right\|^2\right]+\frac{\eta L^2\epsilon_w}{2}(\zeta+4\eta)\nonumber\\
              &\;+\left(1+\frac{\zeta}{2\eta}\right)\mathbb{E}\left[\sum^K_{k=1}\bar D_k^i\left\|\mathbf{w}_k^i-{\bar{\mathbf{w}}^{i}}\right\|^2\right],
\end{align}where we have
\begin{align}\label{thrid_term_upperbound}
&\mathbb{E}\left[\sum^K_{k=1}\bar D_k^i\left\|\mathbf{w}_k^i-{\bar{\mathbf{w}}^{i}}\right\|^2\right]=\sum^K_{k=1}\bar D_k^i\mathbb{E}\Bigg[\Biggl\|\left(\mathbf{w}_k^i-\bar{\mathbf{w}}^{i'}\right) \nonumber\\
&-\left({\bar{\mathbf{w}}^{i}}-\bar{\mathbf{w}}^{i'}\right)\Biggr\|^2\Bigg]\leq\sum^K_{k=1}\bar D_k^i\mathbb{E}\left[\left\|\left(\mathbf{w}_k^i-\bar{\mathbf{w}}^{i'}\right)\right\|^2\right],
\end{align}{where the inequality follows from the fact $\mathbb{E}\left[\|x - \mathbb{E}[x]\|^2\right] \leq \mathbb{E}\left[\|x\|{^2}\right]$ and $\mathbb{E}\left[\mathbf{w}_k^i-\bar{\mathbf{w}}^{i'}\right]=\bar{\mathbf{w}}^i-\bar{\mathbf{w}}^{i'}$. }

Subsequently, the third term in (\ref{A_2_upp}) can be bounded by invoking Assumption \ref{Assumption_sigmaG} and {considering two cases: $i+1\in\mathcal{I}$ and $i+1\not\in\mathcal{I}$.} Specifically, when $i+1\in\mathcal{I}$, we assume that there exists an update step $i'\leq i$ such that $i'\in\mathcal{I}$ and $i-i'\leq I-1$. Then, we have
\begin{align}\label{local_iter_expect_i+1}
&\sum^K_{k=1}\bar D_k^i\mathbb{E}\left[\left\|\left(\mathbf{w}_k^i-\bar{\mathbf{w}}^{i'}\right)\right\|^2\right]\\
&=\sum^K_{k=1}\bar D_k^i\mathbb{E}\left[\left\|-\sum_{t=i'}^{i-1}\eta\nabla f_{L,k}(\mathbf{w}^t_k,s_k^t)\right\|^2\right]\leq\nonumber\\
&\sum^K_{k=1}\bar D_k^i(I-1)\sum_{t=i'}^{i-1}\eta^2\mathbb{E}\left[\left\|\nabla f_{L,k}(\mathbf{w}^t_k,s_k^t)\right\|^2\right]=(I-1)^2\eta^2\epsilon_s^2 \;. \nonumber
\end{align}
On the other hand, when $i+1\not\in\mathcal{I}$ and $i \in \mathcal{I}$, we assume that there exists an update step $i'\leq i$ such that $i'\in\mathcal{I}$ and $i-i'= I$. Then, we have
\begin{align}\label{local_iter_expect_i}
&\sum^K_{k=1}\bar D_k^i\mathbb{E}\left[\left\|\left(\mathbf{w}_k^i-\bar{\mathbf{w}}^{i'}\right)\right\|^2\right]\\
&=\sum^K_{k=1}\bar D_k^i\mathbb{E}\left[\left\|\sum^K_{k=1}\bar D_k^{i-1}\left(-\sum_{t=i'}^{i-1}\eta\nabla f_{L,k}(\mathbf{w}^t_k,s_k^t)\right)\right\|^2\right]\nonumber\\
&\leq\sum^K_{k=1}\bar D_k^i\sum^K_{k=1}\bar D_k^{i-1}I\sum_{t=i'}^{i-1}\eta^2\mathbb{E}\left[\left\|\nabla f_{L,k}(\mathbf{w}^t_k,s_k^t)\right\|^2\right]=I^2\eta^2\epsilon_s^2  \;. \nonumber
\end{align}
By using (\ref{local_iter_expect_i+1}) and (\ref{local_iter_expect_i}) in (\ref{thrid_term_upperbound}), $\mathbb{E}\left[\|A_2\|^2\right]$ in (\ref{A_2_upp}) is then upper bounded by
\begin{align}\label{A_2_upp_final}
     &\mathbb{E}\left[\|A_2\|^2\right]\leq (1-\eta\mu)\mathbb{E}\left[\left\|\bar{\mathbf{w}}^i-\mathbf{w}^*\right\|^2\right]+\frac{\eta L^2\epsilon_w}{2}(\zeta+4\eta)\nonumber\\
              &\;+\left(1+\frac{\zeta}{2\eta}\right) I^2\eta^2\epsilon_s^2 \;.
\end{align}

From the definition of $A_3$ in (\ref{weight_optimal_expression}), we can get 
\begin{align}\label{A_3_expect}
    &\mathbb{E}\left[\|A_3\|^2\right]=\mathbb{E}\left[\left\|\sum^K_{k=1}\bar D_k^i\eta\left(\nabla f_{L,k}(\mathbf{w}_k^i)-\nabla f_{L,k}(\mathbf{w}_k^i,s_k^i)\right)\right\|^2\right]\nonumber\\
    &\leq\sum^K_{k=1}(\bar D_k^i)^2\eta^2\mathbb{E}\left[\left\|\left(\nabla f_{L,k}(\mathbf{w}_k^i)-\nabla f_{L,k}(\mathbf{w}_k^i,s_k^i)\right)\right\|^2\right]\nonumber\\
    &\leq\sum^K_{k=1}(\bar D_k^i)^2\eta^2\epsilon^2_v \;,
\end{align}{where Jensen's inequality and Assumption 3 are applied in the first and second inequalities, respectively.} 

By substituting (\ref{A_2_upp_final}) and (\ref{A_3_expect}) into (\ref{weight_optimal_expression}), we obtain an upper bound on the expected difference:
\begin{align}\label{local_iter_expect_i+1_i}
    &\mathbb{E}\left[\|\bar{\mathbf{w}}^{i+1}-\mathbf{w}^*\|^2\right]\nonumber\\    
    &\leq (1-\eta\mu)\mathbb{E}\left[\left\|\bar{\mathbf{w}}_k^i-\mathbf{w}^*\right\|^2\right]+\frac{\eta L^2\epsilon_w}{2}(\zeta+4\eta)\nonumber\\
    &\;\;\;\;\;+\left(1+\frac{\zeta}{2\eta}\right)I^2\eta^2\epsilon_s^2+\sum^K_{k=1}(\bar D_k^i)^2\eta^2\epsilon^2_v.
\end{align}
By recursively applying the bound in (\ref{local_iter_expect_i+1_i}), we derive the final result given in (\ref{accuracy_proof}). Hence, the proof is completed.

\section{{Proof of Lemma \ref{B_3_lemma}}} \label{B_3_proof}

From (\ref{A_2_upp_1}), the term $B_3$ can be rewritten in (\ref{B_3_exp}) at the {top of the next page}, {where the term $C_1$ is bounded in (\ref{C_1_exp})} by applying a first-order Taylor expansion in step ($a$), invoking Lemma \ref{AMGM_CS} in step ($b$) and leveraging the $L$-smoothness property in step ($c$). Substituting (\ref{C_1_exp}) into (\ref{B_3_exp}) yields an upper bound on \( B_3 \):
\setcounter{equation}{0}
\renewcommand\theequation{C.\arabic{equation}}
\begin{figure*}[t]
\begin{align}
B_3&=4L\eta^2\sum^K_{k=1}\bar D_k^i\left(f_{L,k}(\mathbf{w}_k^i)-f_{L,k}(\mathbf{w}_k^*)\right)-2\eta\sum^K_{k=1}\bar D_k^i\left(f_{L,k}(\mathbf{w}_k^i)-f_{L,k}(\mathbf{w}^*)\right)\nonumber\\
    &=-2\eta(1-\eta2L)\sum^K_{k=1}\bar D_k^i\left(f_{L,k}(\mathbf{w}_k^i)-f_{L,k}(\mathbf{w}^*)\right)+\eta^24L\left(
    \sum^K_{k=1}\bar D_k^i\left(f_{L,k}(\mathbf{w}_k^i)-f_{L,k}(\mathbf{w}_k^*)\right)-\sum^K_{k=1}\bar D_k^i\left(f_{L,k}(\mathbf{w}_k^i)-f_{L,k}(\mathbf{w}^*)\right)\right)\nonumber\\
    &=-2\eta(1-\eta2L)\underbrace{\sum^K_{k=1}\bar D_k^i\left(f_{L,k}(\mathbf{w}_k^i)-f_{L,k}(\mathbf{w}^*)\right)}_{C_1}+\eta^24L\left(\sum^K_{k=1}\bar D_k^i\left(f_{L,k}(\mathbf{w}^*)-f_{L,k}(\mathbf{w}_k^*)\right)\right).\label{B_3_exp}\\
C_1&=\sum^K_{k=1}\bar D_k^i\left(f_{L,k}(\mathbf{w}^i_k)-f_{L,k}\left(\bar{\mathbf{w}}^{i}\right)\right)+\sum^K_{k=1}\bar D_k^i\left(f_{L,k}\left(\bar{\mathbf{w}}^{i}\right)-f_{L,k}(\mathbf{w}^*)\right)\nonumber\\
&\overset{(a)}\geq\sum^K_{k=1}\bar D_k^i{\left\langle\nabla f_{L,k}\left(\bar{\mathbf{w}}^{i}\right),\mathbf{w}^i_k-\bar{\mathbf{w}}^{i}\right\rangle}+\sum^K_{k=1}\bar D_k^i\left(f_{L,k}\left(\bar{\mathbf{w}}^{i}\right)-f_{L,k}(\mathbf{w}^*)\right)\nonumber\\
&\overset{(b)}\geq-\sum^K_{k=1}\bar D_k^i\left(\frac{1}{2}\eta\left\|\nabla f_{L,k}\left(\bar{\mathbf{w}}^{i}\right)\right\|^2+\frac{1}{2\eta}\left\|\mathbf{w}^i_k-\bar{\mathbf{w}}^{i}\right\|^2\right)+\sum^K_{k=1}\bar D_k^i\left(f_{L,k}\left(\bar{\mathbf{w}}^{i}\right)-f_{L,k}(\mathbf{w}^*)\right)\nonumber\\
&\overset{(c)}\geq-\sum^K_{k=1}\bar D_k^i\left(\eta L \left(f_{L,k}\left(\bar{\mathbf{w}}^{i}\right)-f_{L,k}(\mathbf{w}_k^*)\right)+\frac{1}{2\eta}\left\|\mathbf{w}^i_k-\bar{\mathbf{w}}^{i}\right\|^2\right)+\sum^K_{k=1}\bar D_k^i\left(f_{L,k}\left(\bar{\mathbf{w}}^{i}\right)-f_{L,k}(\mathbf{w}^*)\right).\label{C_1_exp}
\end{align}
\end{figure*}
\begin{align}\label{B3_upperbound_1}
    B_3&\leq\zeta\Biggl[\sum^K_{k=1}\bar D_k^i\Biggl(\eta L \left(f_{L,k}\left(\bar{\mathbf{w}}^{i}\right)-f_{L,k}(\mathbf{w}_k^*)\right)\nonumber\\
    &\;+\frac{1}{2\eta}\left\|\mathbf{w}^i_k-\bar{\mathbf{w}}^{i}\right\|^2\Biggr)-\sum^K_{k=1}\bar D_k^i\left(f_{L,k}\left(\bar{\mathbf{w}}^{i}\right)-f_{L,k}(\mathbf{w}^*)\right)\Biggr]\nonumber\\
    &\;+\eta^24L\sum^K_{k=1}\bar D_k^i\left(f_{L,k}(\mathbf{w}^*)-f_{L,k}(\mathbf{w}_k^*)\right)\,,
\end{align}
 {where $\zeta=2\eta(1-\eta2L)\geq 0$. To simplify the expression, we define $\Gamma \triangleq \sum^K_{k=1}\bar D_k^i(f_{L,k}(\mathbf{w}^*)-f_{L,k}(\mathbf{w}_k^*)$. Applying Remark \ref{assumption_lsmooth_remark_2} which states that $f_{L,k}(\mathbf{w})$ is $L$-smooth, we have $\Gamma \leq \frac{L}{2}\sum^K_{k=1}\bar D_k^i(\|\mathbf{w}^*-\mathbf{w}_k^*\|^2)$. Assuming that the deviation between the optimal parameter of each UE and the global optimal parameter is bounded, i.e.,  $\|\mathbf{w}^*-\mathbf{w}_k^*\|^2\leq\epsilon_w$, we obtain $\Gamma\leq \frac{L\epsilon_w}{2}$. Hence, the upper bound in (\ref{B3_upperbound_1}) can be rewritten as ($a$) in (\ref{B_3_exp_fin}). Noting that $\sum^K_{k=1}\bar D_k^i=1$ and $f_{L,k}\left(\bar{\mathbf{w}}^{i}\right) \geq f_{L,k}(\mathbf{w}^*)$, we have $\sum^K_{k=1}\bar D_k^i\left(f_{L,k}\left(\bar{\mathbf{w}}^{i}\right)-\sum^K_{k'=1}\bar D_{k'}^i f_{L,{k'}}(\mathbf{w}^*)\right)= \sum^K_{k=1}\bar D_k^if_{L,k}\left(\bar{\mathbf{w}}^{i}\right)- \sum^K_{k=1}\bar D_k^if_{L,k}(\mathbf{w}^*) \geq 0$. Moreover, since $\eta \leq \frac{1}{2L}$, it implies $(\eta L - 1) \leq 0$. Substituting the above result into ($a$) of (\ref{B_3_exp_fin}) yields ($b$) of (\ref{B_3_exp_fin}). The proof is complete.}
\begin{figure*}[t]
\begin{align}
    B_3&\overset{(a)}\leq\underbrace{\zeta(\eta L-1)\sum^K_{k=1}\bar D_k^i \left(f_{L,k}\left(\bar{\mathbf{w}}^{i}\right)-\sum^K_{k=1}\bar D_k^i f_{L,k}(\mathbf{w}^*)\right)}_{\leq 0}+\eta L\Gamma(\zeta+4\eta)+\frac{\zeta}{2\eta}\sum^K_{k=1}\bar D_k^i\left\|\mathbf{w}^i_k-\bar{\mathbf{w}}^{i}\right\|^2\nonumber\\
   &\overset{(b)}\leq\frac{\eta L^2\epsilon_w}{2}(\zeta+4\eta)+\frac{\zeta}{2\eta}\sum^K_{k=1}\bar D_k^i\left\|\mathbf{w}^i_k-\bar{\mathbf{w}}^{i}\right\|^2.\label{B_3_exp_fin}
\end{align}
\hrule
\end{figure*}

\bibliography{references1}
\bibliographystyle{IEEEtran}

\end{document}